\documentclass[11pt]{article}
\usepackage{graphicx}
\usepackage{latexsym}
\usepackage{amssymb}
\usepackage{amsmath}
\usepackage{amsthm}
\usepackage{forloop}
\usepackage[paper=letterpaper,margin=1in]{geometry}
\usepackage[ruled,vlined,linesnumbered]{algorithm2e}
\usepackage{enumitem}
\usepackage{nicefrac}

\newtheorem{theorem}{Theorem}[section]
\newtheorem{lemma}[theorem]{Lemma}

\newtheorem{observation}[theorem]{Observation}
\newtheorem{claim}[theorem]{Claim}

\makeatletter
\newtheorem*{rep@theorem}{\rep@title}
\newcommand{\newreptheorem}[2]{%
\newenvironment{rep#1}[1]{%
 \def\rep@title{#2 \ref{##1}}%
 \begin{rep@theorem}}%
 {\end{rep@theorem}}}
\makeatother

\newreptheorem{lemma}{Lemma}

\newcommand{\defcal}[1]{\expandafter\newcommand\csname c#1\endcsname{{\mathcal{#1}}}}
\newcommand{\defbb}[1]{\expandafter\newcommand\csname b#1\endcsname{{\mathbb{#1}}}}
\newcounter{calBbCounter}
\forLoop{1}{26}{calBbCounter}{
    \edef\letter{\Alph{calBbCounter}}
    \expandafter\defcal\letter
		\expandafter\defbb\letter
}

\newcommand{\eps}{\varepsilon}
\newcommand{\ie}{{\it i.e.}}
\newcommand{\eg}{{\it e.g.}}

\newcommand{\characteristic}{{\mathbf{1}}}
\DeclareMathOperator{\support}{supp}

\DeclareMathOperator*{\argmax}{arg\;max}

\begin{document}

\title{{\bf Deterministic Algorithms for Submodular Maximization Problems}}

\author{
Niv Buchbinder\thanks{Department of Statistics and Operations Research, Tel Aviv University, Israel.  e-mail: niv.buchbinder@gmail.com}
\and
Moran Feldman\thanks{School of Computer and Communication Sciences, EPFL, Switzerland. e-mail: moran.feldman@epfl.ch.}
}

\maketitle

\begin{abstract}
Randomization is a fundamental tool used in many theoretical and practical areas of computer science. We study here the role of randomization in the area of submodular function maximization.
In this area most algorithms are randomized, and in almost all cases the approximation ratios obtained by current randomized algorithms are superior to the best results obtained by known deterministic algorithms. Derandomization of algorithms for general submodular function maximization seems hard since the access to the function is done via a value oracle. This makes it hard, for example, to apply standard derandomization techniques such as conditional expectations.
Therefore, an interesting fundamental problem in this area is whether randomization is inherently necessary for obtaining good approximation ratios.

In this work we give evidence that randomization is not necessary for obtaining good algorithms by presenting a new technique for derandomization of algorithms for submodular function maximization.
Our high level idea is to maintain explicitly a (small) distribution over the states of the algorithm, and carefully update it using marginal values obtained from an extreme point solution of a suitable linear formulation.
We demonstrate our technique on two recent algorithms for unconstrained submodular maximization and for maximizing submodular function subject to a cardinality constraint. In particular, for unconstrained submodular maximization we obtain an optimal deterministic $1/2$-approximation showing that randomization is unnecessary for obtaining optimal results for this setting.
\end{abstract}

\thispagestyle{empty}
\setcounter{page}{0}
\newpage

\pagenumbering{arabic}

\section{Introduction}
Randomization is a fundamental tool used in many theoretical and practical areas of computer science \cite{AS20,MU05}. It is widely used, for example, in cryptology, coding, and the design of approximation and online algorithms.
In the area of approximation algorithms randomization is used both to improve the running time of algorithms and to obtain improved approximation ratios. However, in the majority of cases randomization can be removed to obtain an equivalent deterministic algorithm with the same approximation ratio (albeit, in some cases, with a longer running time).
One central field of combinatorial optimization in which it is unclear whether it is possible to avoid randomization is the field of submodular function maximization.

A set function $f:2^{\mathcal{N}}\rightarrow \mathbb{R}$ is submodular if for every $A,B\in \mathcal{N}$:
$f(A\cap B)+f(A\cup B)\leq f(A) +f(B)$. An equivalent definition of submodularity, which is perhaps more intuitive, is that of diminishing returns:
$ f(A\cup \{ u\}) - f(A) \geq f(B\cup \{ u\}) - f(B)$ for every $A\subseteq B\subseteq \mathcal{N}$ and $u\notin B$.
The concept of diminishing returns is widely used in economics, and thus, it should come as no surprise that utility functions in economics are often submodular. Additionally, submodular functions are ubiquitous in various other disciplines, including combinatorics, optimization, information theory,
operations research, algorithmic game theory and machine learning. A few well known examples of submodular functions from these disciplines include cut functions of graphs and hypergraphs,
rank functions of matroids, entropy, mutual information, coverage functions and budget additive functions. Moreover, submodular maximization problems capture well known combinatorial optimization problems such as: Max-Cut \cite{GW95,H01,K72,KKMO07,TSSW00}, Max-DiCut \cite{FG95,GW95,HZ01},
Generalized Assignment \cite{CK05,CKR06,FV06,FGMS06} and Max-Facility-Location \cite{AS99,CFN77a,CFN77b}.

For a general submodular function the explicit representation of the function might be exponential in the size of its ground set.
Hence, it is standard practice to assume that the function is accessed via a value oracle returning the value of $f(S)$ when given a set $S \subseteq \cN$.
The use of a value oracle makes it difficult to use the standard derandomization method of conditional expectations (see, \eg, \cite{AS20}). Moreover, other standard methods, such as using a small sample space, seem to fail as well. Indeed, for most scenarios the best current approximation algorithms are randomized, while known deterministic algorithms achieve only inferior approximation ratios.
One example is the problem of unconstrained submodular maximization~\cite{BFNS12} for which the best (optimal) randomized algorithm has an approximation ratio of $1/2$, while the approximation ratio of the best known deterministic algorithm is only $1/3$. Another example is maximizing a monotone submodular function subject to a matroid constraints~\cite{FNW78,CCPV11,FNS11}. For this problem the best deterministic algorithm has an approximation ratio of $1/2$, while the best (optimal) randomized algorithm has an improved approximation ratio of $1-1/e$. An interesting fundamental problem is whether randomization is inherently necessary for obtaining good approximation ratio in the field of submodular function maximization.

\subsection{Our results}
In this paper we give evidence that randomization is not necessary for obtaining good algorithms for submodular maximization. We present a new technique for derandomization of algorithms for submodular maximization based on the following idea. In a typical randomized algorithm the size of the support of the distribution implied by the algorithm is exponential (as otherwise it can often be trivially enumerated).
We show that in certain cases the size of the distribution can be kept small (polynomial) throughout the execution of the algorithm. This is done by formulating the set of ``good" updates to the distribution in each iteration of the algorithm as a linear formulation, and then choosing an update step that corresponds on an (optimal) {\bf extreme} point of the linear program. We then use the fact that an extreme point for our formulations does not have many non-integral variables to control the increase in the size of the support of the distribution. This allows us to maintain the distribution explicitly throughout the execution. We are not aware of any previous results that obtain derandomization via such an approach, and believe that this idea may be applicable for other settings as well.

We demonstrate our technique on two recent algorithms. The first of them is a randomized $1/2$-approximation algorithm presented by~\cite{BFNS12} for the problem of unconstrained submodular maximization. It is known that the approximation ratio of the last algorithm cannot be improved by any polynomial time algorithm~\cite{FMV11}. Using our technique we obtain the following result. In this result, and throughout the rest of the paper, we use $n$ to denote the size of the ground set $\cN$.

\begin{theorem} \label{thm:unconstrained}
Let $f$ be a submodular function.
There exists a {\bf deterministic} algorithm with an approximation ratio of $1/2$ for the problem $\max_{S\subseteq \cN}\{f(S)\}$. The algorithm makes $O(n^2)$ value oracle queries.
\end{theorem}

Notice that this result shows that randomization is not necessary for obtaining the best possible approximation ratio for the problem. However, this comes at a cost, as the number of oracle queries made by our deterministic algorithm is $O(n^2)$, while the randomized algorithm only needs $O(n)$ oracle queries.

Our second result is for maximizing a (non-monotone) submodular function subject to a cardinality constraint. For this problem the best known randomized algorithm has an approximation guarantee of $1/e+0.004$~\cite{BFNS14}. This algorithm is based on a combination of two randomized algorithms; one of which is an elegant randomized greedy algorithm that has approximation ratio $1/e$ (on its own). We show that one can obtain an equivalent deterministic algorithm.

\begin{theorem} \label{thm:cardinality}
Let $f$ be a submodular function.
There exists a {\bf deterministic} algorithm that has an approximation ratio of $1/e$ for the problem of $\max_{|S|\leq k}\{f(S)\}$. The algorithm makes $O(k^2 n)$ value oracle queries.
\end{theorem}

Again, note that our deterministic algorithm makes $O(k^2 n)$ oracle queries, while the randomized algorithm only needs $O(kn)$ queries.

\subsection{Previous Results}
The literature on submodular maximization problems is very large, and therefore, we mention below only a few of the most relevant works.
Randomization is widely used in submodular maximization. In particular, many of the recent algorithms use an extension of submodular functions to fractional vectors known as the \emph{multilinear} extension (see, \eg, \cite{CCPV11,GV11,FNS11,KST13,BFNS14}). Any algorithm using this extension must be randomized since the only known way to (approximately) evaluate this extension is via random sampling. A few examples of randomized algorithms for submodular maximization that do not use the multilinear extension can be found in~\cite{DS06,AGR11,BFNS12,FW14}.

The first provable approximation algorithms for unconstrained submodular maximization were described by Feige et al.~\cite{FMV11}. Their best algorithms for the problem achieve a randomized approximation ratio of $0.4 - o(1)$, and a deterministic approximation ratio of $1/3 - \eps$ (where $\eps$ is an arbitrarily small positive constant). On the negative side, \cite{FMV11} showed that no algorithm has an approximation ratio of $1/2 + \eps$ for the problem. The randomized approximation ratio was improved gradually~\cite{GV11,FNS11b}, eventually leading to an optimal linear time $1/2$-approximation randomized algorithm given by~\cite{BFNS12}. However, the deterministic approximation ratio has not been improved since the work of~\cite{FMV11}. Interestingly,~\cite{BFNS12} described a different $1/3$-approximation deterministic algorithm for the problem, which led some to conjecture that no deterministic algorithm can do better. Huang and Borodin~\cite{HB14} strengthened this conjecture by showing that a large class of deterministic algorithms resembling the optimal $1/2$-approximation randomized algorithm of~\cite{BFNS12} cannot achieve $1/2$-approximation for unconstrained submodular maximization.

The problem of maximizing a (non-monotone) submodular function subject to a matroid independence constraint (which generalizes a cardinality constraint) was given a deterministic $(1/4 - \eps)$-approximation algorithm by Lee et al.~\cite{LMNS10} and a randomized $0.309$-approximation algorithm by Vondr\'{a}k~\cite{V13}. Later, the randomized approximation ratio was improved to $0.325$ using a simulated annealing technique~\cite{GV11}, and then to $1/e-o(1)$~\cite{FNS11} via an extension of the continuous greedy algorithm of~\cite{CCPV11}. All the above randomized results are based on the multilinear extension, and thus, are quite inefficient. Buchbinder et al.~\cite{BFNS14} described a simple randomized greedy algorithm designed specifically for a cardinality constraint, and used this algorithm to achieve a clean approximation ratio of $1/e$ with a significantly better time complexity. Additionally, \cite{BFNS14} also showed that $1/e$ is not the correct approximation ratio for the case of a cardinality constraint by describing an (inefficient) polynomial time $(1/e + 0.004)$-approximation algorithm for it. On the hardness side, \cite{GV11} showed that no polynomial time algorithm can have an approximation ratio better than $0.491$ for the problem of maximizing a (non-monotone) submodular function subject to a cardinality constraint.

Recent works consider online and streaming variants of the above problems~\cite{BFS15b,CGQ15} as well as faster algorithms~\cite{BFS15}.

\section{Additional Notation}

For every set $S$ and element $u$, we denote the union $S \cup \{u\}$ by $S + u$, and the expression $S \setminus \{u\}$ by $S - u$. Given a submodular function $f : 2^\cN \rightarrow \mathbb{R}$, the marginal contribution of $u$ to $S$ is denoted by $f(u \mid S) = f(S + u) - f(S)$.
Our algorithms {\bf explicitly} maintain in each iteration $i$ a (finite) distribution $\cD_i$ over possible states of the algorithm. Each distribution is represented as a multiset of tuples $\{(p, S)\}$, where $S$ is a state and $p$ is the probability of this state. Naturally, we require all the $p$-values to be positive and to add up to $1$. We denote by $|\cD_i|$ the number of tuples in the distribution $\cD_i$, and by $\support(\cD_{i})$ the set of states represented by these tuples (which is also the set of states having a positive probability).

To simplify the presentation of our algorithms, we allow our distributions to contain multiple tuples with the same state. Moreover, there might even be multiple identical tuples (this is why the distributions have to be multisets). Whenever this happens, the meaning is that the probability of a state $S$ is the sum of the $p$-values of tuples containing it. An implementation of our algorithm can either keep identical state tuples separate, or unify them. Our proofs are independent of such details. The pseudocode of our algorithms uses the first option, which requires us to assume the following semantic rules regarding multisets:
\begin{itemize}
	\item Given multisets $A$ and $B$ of tuples, the multiplicity of a tuple in the union $A \cup B$ is the sum of its multiplicities in the original sets.
	\item Given an expression of the from $\{(p(x), S(x)) \mid x \in A\}$, where $A$ is a set and $(p(x), S(x))$ is a tuple which is a function of an element $x \in A$. If there are multiple $x$ values resulting a single tuple $(p, S)$, then we assume that the multiplicity of this tuple is equal to the number of such $x$ values in $A$.
\end{itemize}

\section{Unconstrained Submodular Maximization}
In this section we present a deterministic $1/2$-approximation algorithm for the problem $\max\{f(S) : S\subseteq \cN\}$ whose formal description is given as Algorithm~\ref{alg:doub-greedy}. Each state in the distribution maintained by our algorithm is a pair $(X, Y)$ of sets. Technically, this means that the distribution should be a multiset of tuples $(p, (X, Y))$. However, we abuse notation and use tuples of the form $(p, X, Y)$ instead. Additionally, we write $\bE_{\cD_i}$ instead of $\bE_{(X, Y) \sim \cD_i}$ to avoid visual cluttering.

\begin{algorithm*}[!h]
\caption{\textsf{Deterministic Unconstrained}$(f)$} \label{alg:doub-greedy}
Initialize: $\cD_0=\{(1, \varnothing, \cN )\}$.\\
Denote the elements of $\cN$ by $u_1, u_2, \dotsc, u_n$, in an arbitrary order (recall that $n = |\cN|$).\\
\For{$i$ = $1$ \KwTo $n$}
{
    $\forall (X,Y)\in \support (\cD_{i-1})$, let $a_{i}(X)= f(X +u_i)-f(X)$, $b_{i}(Y)=f(Y - u_i)-f(Y)$.\\
    Find an {\bf extreme point} solution of the following linear formulation:
\[\begin{array}{lllll}
(P) & \bE_{\cD_{i - 1}}[z(X,Y) a_{i}(X) + w(X,Y) b_{i}(Y)] & \geq  2 \cdot \bE_{\cD_{i - 1}}[z(X,Y) b_{i}(Y)] \mspace{-240mu} & \\
& \bE_{\cD_{i - 1}}[z(X,Y) a_{i}(X) + w(X,Y) b_{i}(Y)] & \geq  2 \cdot \bE_{\cD_{i - 1}}[w(X,Y) a_{i}(X)]  \mspace{-190mu} &\\
&  z(X,Y) + w(X,Y) & = 1 & \forall (X,Y) \in \support(\cD_{i-1})\\
& z(X,Y),w(X,Y) & \geq 0 & \forall (X,Y) \in \support(\cD_{i-1})
\end{array}\]\\
Construct a new distribution: \label{line:construction}
\begin{align*}
	\cD_{i} \gets{} & \{(z(X,Y) \cdot {\textstyle \Pr_{\cD_{i - 1}}[(X,Y)]}, (X+u_i,Y)) \mid (X,Y) \in \support(\cD_{i-1}), z(X,Y) > 0\}\\
					& \cup \{(w(X,Y) \cdot {\textstyle \Pr_{\cD_{i - 1}}[(X,Y)]}, (X,Y-u_i)) \mid (X,Y) \in \support(\cD_{i-1}), w(X,Y) > 0\}
					\enspace.
\end{align*}
}
\Return{$\argmax_{(X,Y) \in \support(\cD_n)}\{f(X)\}$} (equivalent to $\argmax_{(X,Y) \in \support(\cD_n)}\{f(Y)\}$).
\end{algorithm*}

We begin the analysis of Algorithm~\ref{alg:doub-greedy} with the following simple observations.

\begin{observation} \label{obs:simple_unconstrained}
The following holds for every iteration $1 \leq i \leq n$ of Algorithm~\ref{alg:doub-greedy}:
\begin{enumerate}
\item For every state $(X,Y)\in \support(\cD_i)$, $X \cap \{u_{i + 1}, \ldots, u_n\} = \varnothing$, $\{u_{i + 1}, \ldots, u_n\} \subseteq Y$ and $X \cap \{u_{1}, \ldots, u_i\} = Y \cap \{u_{1}, \ldots, u_i\}$. \label{obs:valid_distribution21}
\item The total sum of the probabilities in $\cD_i$ is $1$, and thus, $\cD_i$ is a valid distribution. \label{obs:valid_distribution23}
\item The formulation $(P)$ is feasible. In particular, one feasible solution assigns for every state $(X,Y) \in \support(\cD_{i - 1})$:
\[
	z(X,Y) = \frac{\max\{0, a_i(X)\}}{\max\{0, a_i(X)\}+\max\{0, b_i(Y)\}}
	\quad \text{(or $1$ if the denominator is $0$)}
\]
and
\[
	w(X,Y) = 1-z(X,Y) = \frac{\max\{0, b_i(Y)\}}{\max\{0, a_i(X)\}+\max\{0, b_i(Y)\}}
	\quad \text{(or $0$ if the denominator is $0$)}
	\enspace.
\]
\label{obs:valid_assignment22}
\item For any extreme point of $(P)$ there are at most $2+|\cD_{i-1}|$ non-zero variables. Thus, $|\cD_i| \leq 2+|\cD_{i-1}|$, and $|\cD_n| \leq 2n + 1$. \label{obs:size_increase24}
\end{enumerate}
\end{observation}

\noindent \textbf{Remark:} Item~\ref{obs:size_increase24} of Observation~\ref{obs:simple_unconstrained} is the only place in the analysis of Algorithm~\ref{alg:doub-greedy} where we use the fact that the algorithm finds an extreme point solution of $(P)$.  Claim~\ref{cl-unconstraind-extreme} shows that one can in fact compute a solution for $(P)$ having at most $1+|\cD_{i-1}|$ non-zero variables. Moreover, it is possible to compute such a solution in near linear time, implying that implementation of Algorithm~\ref{alg:doub-greedy} does not require an LP solver.

\begin{proof}
The proof of the observation is by induction on $i$. Assume the observation holds for every $1 \leq i' < i$, and let us prove it for $i$.
Observe that item~(\ref{obs:valid_distribution21}) holds for $i - 1$ either by the induction hypothesis or (for $i = 1$) by the fact that $\cD_0$ trivially obeys it. Additionally, it is easy to see that every state of $\cD_i$ is obtained from a state of $\cD_{i - 1}$ by either adding $u_i$ to $X$ or removing $u_i$ from $Y$. Hence item~(\ref{obs:valid_distribution21}) is maintained.
Similarly, item~(\ref{obs:valid_distribution23}) holds for $i - 1$. Since $z(X,Y) + w(X,Y) = 1$ for every state $(X,Y)\in \support(\cD_{i-1})$, the distribution $\cD_i$ contains up two tuples resulting from $(X, Y)$, and these tuples exactly split the probability $(X, Y)$ has in $\cD_{i - 1}$. Thus, the sum of the probabilities in $\cD_i$ is equal to their sum in $\cD_{i - 1}$.

To see why item~(\ref{obs:valid_assignment22}) holds, observe first that by~(\ref{obs:valid_distribution21}) $X \subseteq Y\setminus \{u_i\}$ for each $(X,Y)\in \support(\cD_{i-1})$, and therefore, by submodularity,
\[a_i(X) + b_i(Y) = [f(X +u_i)-f(X)] + [f(Y - u_i)-f(Y)]\geq 0 \enspace.\]

Next, we prove that the first constraint of $(P)$ is satisfied by the assignment suggested by~(\ref{obs:valid_assignment22}). 
Proving that the second constraint of $(P)$ is satisfied by the assignment can be done similarly.

If $a_i(X)=b_i(Y)=0$ for some state $(X, Y)$, then $(X,Y)$ does not contribute to either side of the first constraint of $(P)$, and we may ignore it.
Thus, we may assume that either $a_i(X)$ or $b_i(Y)$ is strictly non-zero for every state. Plugging the assignment suggested by~(\ref{obs:valid_assignment22}) under this assumption into the first constraint in $(P)$, we get:
\begin{align*}
& \bE_{\cD_{i - 1}}[z(X,Y)\cdot a_i(X) + w(X,Y)\cdot b_i(Y)] \\
& =  \bE_{\cD_{i - 1}}\left[\frac{\max\{0, a_i(X)\}\cdot a_i(X)}{\max\{0, a_i(X)\}+\max\{0, b_i(Y)\}} + \frac{\max\{0, b_i(Y)\}\cdot b_i(Y)}{\max\{0, a_i(X)\}+\max\{0, b_i(Y)\}}\right]\\
& \geq 2 \cdot \bE_{\cD_{i - 1}}\left[\frac{\max\{0, a_i(X)\}\cdot  b_i(Y)}{\max\{0, a_i(X)\}+\max\{0, b_i(Y)\}}\right] = 2 \cdot \bE_{\cD_{i - 1}}[z(X,Y) \cdot b_i(Y)]
\enspace.
\end{align*}
The final inequality holds even without the expectation due to the following argument: if either $a_i(X)<0$ or $b_i(Y)<0$, then the LHS is non-negative while the RHS is non-positive.
On the other hand, if $a_i(X)\geq 0$ and $b_i(Y)\geq 0$ the inequality reduces to $a^2_i(X) + b^2_i(Y) \geq 2 a_i(X) b_i(Y)$, which clearly holds. 

Item~(\ref{obs:size_increase24}) follows immediately by the properties of an extreme point. Since $(P)$ has $2+|\cD_{i-1}|$ constraints, an extreme point of $(P)$ has at most $2+|\cD_{i-1}|$ non-zero variables. Since a single tuple is added to $\cD_i$ for every non-zero variable, the size of $\cD_i$ is upper bounded by $2+|\cD_{i-1}|$. 
\end{proof}

Let $OPT\subseteq \cN$ be the optimal solution for the problem $\{f(S) : S \subseteq \cN\}$ that we want to approximate, and let $OPT(X, Y)$ be a shorthand for the set $((OPT \cup X) \cap Y)$. The following is the main lemma we need in order to analyze Algorithm~\ref{alg:doub-greedy}.

\begin{lemma}\label{main-lem-double}
For any iteration $1 \leq i \leq n$ of Algorithm~\ref{alg:doub-greedy},
\[
{\textstyle \bE_{\cD_i}}[f(X)+f(Y)]-{\textstyle \bE_{\cD_{i - 1}}}[f(X)+f(Y)]
\geq 2 \cdot \left({\textstyle \bE_{\cD_{i - 1}}}[f(OPT(X, Y))] - {\textstyle \bE_{\cD_{i}}}[f(OPT(X, Y))]\right)
	\enspace.
\]
\end{lemma}
\begin{proof}
Observe that whenever $u_i \notin OPT$:
\begin{align*}
{\textstyle \bE_{\cD_{i - 1}}}[f(OPT(X, Y))] - {\textstyle \bE_{\cD_{i}}}&[f(OPT(X, Y))]\\
& = {\textstyle \bE_{ \cD_{i - 1}}}[z(X,Y) \cdot \left(f(OPT(X, Y)) - f(OPT(X + u_i, Y))\right)]\\
& \leq {\textstyle \bE_{\cD_{i-1}}}[z(X,Y) \cdot \left(f(Y - u_i)-f(Y)\right)]
= {\textstyle \bE_{\cD_{i-1}}}[z(X,Y) \cdot b_i(Y)]
\enspace,
\end{align*}
where the inequality follows by submodularity since $OPT(X, Y) \subseteq Y - u_i$. A similar argument can be used to show that whenever $u_i \in OPT$:
\[
	{\textstyle \bE_{\cD_{i - 1}}}[f(OPT(X, Y))] - {\textstyle \bE_{\cD_{i}}}[f(OPT(X, Y))]
	\leq
	{\textstyle \bE_{\cD_{i-1}}}[w(X,Y) \cdot a_i(X)]
	\enspace.
\]

The lemma now follows by combing the above observations with the next inequality:
\begin{align*}
{\textstyle \bE_{\cD_i}}[f(X)+f(Y)&]-{\textstyle \bE_{\cD_{i - 1}}}[f(X)+f(Y)] \\
={} & {\textstyle \bE_{\cD_{i - 1}}}[z(X,Y) \cdot \left(f(X +u_i)-f(X)\right) + w(X,Y) \cdot \left(f(Y - u_i)-f(Y)\right)] \\
={} & {\textstyle \bE_{\cD_{i - 1}}}[\left(z(X,Y) \cdot a_i(X) + w(X,Y) \cdot b_i(Y)\right)] \\
\geq{} & 2 \cdot \max \left\{{\textstyle \bE_{\cD_{i-1}}}[z(X,Y) \cdot b_i(Y)],  {\textstyle \bE_{\cD_{i-1}}}[w(X,Y) \cdot a_i(X)]\right\}
\enspace,
\end{align*}
where the inequality follows by the constraints of $(P)$.
\end{proof}

We can now prove the next theorem, which implies Theorem~\ref{thm:unconstrained}.

\begin{theorem}
Algorithm~\ref{alg:doub-greedy} is a $\frac{1}{2}$-approximation algorithm performing $O(n^2)$ value oracle queries.
\end{theorem}

\begin{proof}
Adding up Lemma \ref{main-lem-double} over $1 \leq i \leq n$ we get:
\[
	{\textstyle \bE_{\cD_n}}[f(X)+f(Y)]- {\textstyle \bE_{\cD_0}}[f(X)+f(Y)] \geq 2 \cdot \left({\textstyle \bE_{\cD_0}}[f(OPT(X, Y))] - {\textstyle \bE_{\cD_n}}[f(OPT(X, Y))]\right)
	\enspace.
\]
The single state in the support of $\cD_0$ is $(\varnothing, \cN)$. Hence, ${\textstyle \bE_{\cD_0}}[f(OPT(X, Y))] = f(OPT(\varnothing, \cN)) = f(OPT)$. On the other hand, for every state $(X_n,Y_n) \in \support(\cD_{n})$ we have $X_n=Y_n$ by Observation~\ref{obs:simple_unconstrained}, and thus, $OPT(X_n, Y_n) = X_n=Y_n$. Plugging all these observations into the last inequality gives:
\[
	2 \cdot {\textstyle \bE_{\cD_n}}[f(X)]- \left(f(\varnothing)+f(\cN)\right) \geq 2 \cdot \left(f(OPT) - {\textstyle \bE_{\cD_n}}[f(X)]\right)
	\enspace.
\]
Using an averaging argument and the non-negativity of $f$ we now get:
\[
	\argmax_{(X,Y) \in \support(\cD_n)}\{f(X)\}
	\geq
	{\textstyle \bE_{\cD_n}}[f(X)]
	\geq
	\frac{f(OPT)}{2} + \frac{f(\varnothing)+f(\cN)}{4}
	\geq
	\frac{f(OPT)}{2}
	\enspace.
\]

Observation~\ref{obs:simple_unconstrained} shows that $|\cD_{i}|\leq 2i+1$ for every $1 \leq i \leq n$. Since Algorithm~\ref{alg:doub-greedy} performs $2$ oracle queries for every state in $\support(\cD_{i - 1})$, the number of such queries done during the $i$-th iteration is at most $4i -2$. Adding up the last bound over all iterations we get a bound $O(n^2)$ on the total number of oracle queries made by the algorithm.
\end{proof} 
\section{Cardinality Constraints}

In this section we present a deterministic $1/e$-approximation algorithm for the problem $\max\{f(S) : |S| \leq k\}$ whose formal description is given as Algorithm~\ref{alg:DetRandGreedy}. Each state in the distribution maintained by our algorithm is a set $S$.



\begin{algorithm*}[!t]
\caption{\textsf{Deterministic Cardinality}$(f, k)$} \label{alg:DetRandGreedy}
Initialize: $\cD_0=\{(1, \varnothing)\}$.\\
\For{$i$ = $1$ \KwTo $k$}
{
    Let $M_i\subseteq\cN$ be a subset of at most $k$ elements maximizing $\sum_{u \in M_i} \bE_{S \sim \cD_{i - 1}}[f(u \mid S)]$.\\
    Find an {\bf optimal extreme point} solution of the following linear formulation:
\[\begin{array}{lllll}
	(P) & \max & \sum_{u\in M_i} \bE_{S \sim \cD_{i - 1}}[x(u, S) \cdot f(u \mid S)] \mspace{-72mu}&& \\
	&& \bE_{S \sim \cD_{i - 1}}[x(u, S)] & \leq \nicefrac{1}{k} \cdot \Pr_{S \sim \cD_{i - 1}}[u \not \in S] & \forall u\in M_i\\
&& \sum_{u\in M_i} x(u, S) + \ell(S) & = 1 & \forall S \in \support(\cD_{i-1}) \\
&&x(u, S), \ell(S) & \geq 0 & \forall u \in M_i, S \in \support(\cD_{i-1})
\end{array}\]\\
Construct a new distribution:
\begin{align*}
	\cD_{i} \gets{} & \{(x(u, S) \cdot {\textstyle \Pr_{\cD_{i - 1}}[S]}, S+u) \mid u \in M_i, S \in \support(\cD_{i-1}), x(u, S) > 0\}\\
					& \cup \{(\ell(S) \cdot {\textstyle \Pr_{\cD_{i - 1}}[S]}, S) \mid S \in \support(\cD_{i-1}), \ell(S) > 0\}
					\enspace.
\end{align*}
}
Return $\argmax_{S \in \support(\cD_k)}\{f(S_k)\}$.
\end{algorithm*}

We first make the following simple observations.

\begin{observation} \label{obs:simple}
The following holds for every iteration $i=1, \dotsc, k$ of Algorithm~\ref{alg:DetRandGreedy}:
\begin{enumerate}
\item The assignment $x(u, S) = \nicefrac{1}{k} \cdot \characteristic[u \not \in S]$ and $\ell(S) = 1 - |M_i \setminus S| / k$ for every $S \in \support(\cD_{i - 1})$ and $u \in M_i$ is a feasible assignment for the formulation $(P)$. \label{obs:valid_assignment}
\item The total sum of the probabilities in $\cD_i$ is $1$, and thus, $\cD_i$ is a valid distribution. \label{obs:valid_distribution}
\item $|\cD_i| \leq k+|\cD_{i-1}|$. Thus, $|\cD_k| \leq k^2 + 1$. \label{obs:size_increase}
\end{enumerate}
\end{observation}

\begin{proof}
The proof of the observation is by induction on $i$. Assume the observation holds for every $1 \leq i' < i$, and let us prove it for $i$. It is easy to verify that item~\eqref{obs:valid_assignment} holds given that $\cD_{i - 1}$ is a valid distribution. To see why item~\eqref{obs:valid_distribution} holds, observe that the sum of the probabilities in $\cD_i$ is:
\[
	\sum_{S \in \support(\cD_{i-1})} \mspace{-18mu} {\textstyle \Pr_{\cD_{i - 1}}[S]} \cdot \left[\sum_{u \in M_i} x(u, S) + \ell(S) \right]
	=
	\sum_{S \in \support(\cD_{i-1})} {\textstyle \Pr_{\cD_{i - 1}}[S]}
	=
	1
	\enspace.
\]
Finally, to prove item~\eqref{obs:size_increase} notice that the number of constraints in $(P)$ at iteration $i$ is at most $k + |\support(\cD_{i - 1})| \leq k+|\cD_{i-1}|$. By the properties of extreme point solutions the total number of variables that are strictly greater than zero is upper bounded by the number of (tight) constraints. Since a single set is added to $\cD_i$ for every non-zero $x(u, S)$ or $\ell(S)$ variable, the size of $\cD_i$ is also upper bounded by $k+|\cD_{i-1}|$.
\end{proof}

The next lemma upper bounds the probability of an item to be in a set chosen according to the distributions defined by Algorithm~\ref{alg:DetRandGreedy}.

\begin{lemma}\label{lem:upper-prob}
For every element $u \in \cN$ and $0 \leq i \leq k$:
\[
	{\textstyle \Pr_{S \sim \cD_i}}[u \not \in S] \geq \left(1-\frac{1}{k}\right)^i
	\enspace.
\]
\end{lemma}
\begin{proof}
The proof of the lemma is by induction on $i$. The distribution $\cD_0$ gives a probability $1$ to the empty set, and thus, $\Pr_{S \sim \cD_0}[u \not \in S] = 1$ for every $u \in \cN$, \ie, the base case $i = 0$ holds. Next, assume the lemma holds for $0 \leq i - 1$, and let us prove it for $i$.

For simplicity of notation, let us define $x(u, S)$ to be $0$ for every $u \not \in M_i$. A set $S \in \support(\cD_i)$ contains the element $u$ in two cases: if it is constructed from a set in the support of $\cD_{i - 1}$ that contains $u$, or it is constructed by adding $u$ to a set in the support of $\cD_{i - 1}$. Using this observation we get the bound:
\begin{align*}
	{\textstyle \Pr_{S \sim \cD_i}}[u \not \in S]
	={} &
	\sum_{\substack{S \in \support(\cD_{i - 1}) \\ u \not \in S}} \mspace{-18mu} {\textstyle \Pr_{\cD_{i - 1}}[S]} \cdot \left[\sum_{u' \in M_i - u} x(u', S) + \ell(S) \right]\\
	\geq{} &
	\sum_{\substack{S \in \support(\cD_{i - 1}) \\ u \not \in S}} \mspace{-18mu} {\textstyle \Pr_{\cD_{i - 1}}[S]} \cdot \left[\sum_{u' \in M_i} x(u', S) + \ell(S) \right] - \sum_{S \in \support(\cD_{i - 1})} \mspace{-18mu} x(u, S) \cdot {\textstyle \Pr_{\cD_{i - 1}}[S]}\\
	={} &
	{\textstyle \Pr_{S \sim \cD_{i - 1}}[u \not \in S]} - {\textstyle \bE_{S \sim \cD_{i - 1}}[x(u, S)]}
	\geq
	(1 - \nicefrac{1}{k}) \cdot {\textstyle \Pr_{S \sim \cD_{i - 1}}[u \not \in S]}
	\geq
	\left(1-\frac{1}{k}\right)^i
	\enspace,
\end{align*}
where the second equality holds by the second constraint of $(P)$, the second inequality holds by the first constraint of $(P)$ and the last inequality holds by the induction hypothesis.
\end{proof}

The following lemma is an immediate implication of Lemma~2.2 of~\cite{BFNS14}. For completeness, we give an independent proof of it Appendix~\ref{app:omitted}.

\begin{lemma}\label{lem:decrease}
For any subset $T$ and a distribution $\cD$,
\[
	{\textstyle \bE_{S \sim \cD}}[f(T \cup S)] \geq f(S) \cdot \min_{u \in N} {\textstyle \Pr_{S \sim \cD}}[u \not \in S]
	\enspace.
\]
\end{lemma}

The next lemma is the last component we need in order to analyze Algorithm~\ref{alg:DetRandGreedy}.

\begin{lemma}\label{lem:improve}
For any iteration $1 \leq i \leq k$ of Algorithm~\ref{alg:DetRandGreedy},
\[
	{\textstyle \bE_{S \sim \cD_i}}[f(S)]-{\textstyle \bE_{S \sim \cD_{i - 1}}}[f(S)] \geq \nicefrac{1}{k} \cdot {\textstyle \bE_{S \sim \cD_{i - 1}}}[f(OPT \cup S)-f(S)]
	\enspace.
\]
\end{lemma}
\begin{proof}
One can view the construction of $\cD_i$ in the following way: the probability of every set $S \in \support(\cD_{i - 1})$ is split. A fraction of $\ell(S)$ of this probability is kept for $S$, and for every $u \in M_i$ a fraction of $x(u, S)$ of this probability is transferred to $S + u$. Using this view we get:
\begin{align*}
	{\textstyle \bE_{S \sim \cD_i}}[f(S)]-{\textstyle \bE_{S \sim \cD_{i - 1}}}[f(S)]
	={} &
	\sum_{u\in M_i}{\textstyle \bE_{S \sim \cD_{i - 1}}}[x(u, S) \cdot f(u \mid S)]
	\geq
	\frac{1}{k} \cdot \sum_{u \in M_i}{\textstyle \bE_{S \sim \cD_{i - 1}}}[f(u \mid S)]\\
	\geq{} &
	\frac{1}{k} \cdot \sum_{u \in OPT}{\textstyle \bE_{S \sim \cD_{i - 1}}}[f(u \mid S)]
	\geq
	\frac{1}{k} \cdot {\textstyle \bE_{S \sim \cD_{i - 1}}}[f(OPT \cup S) - f(S)]
	\enspace,
\end{align*}
where the first inequality holds since the solution found by Algorithm~\ref{alg:DetRandGreedy} must be at least as good as the feasible solution given by Observation~\ref{obs:simple}, the second inequality holds by the definition of $M_i$ and the last inequality holds by submodularity.
\end{proof}

We can now prove the next theorem, which implies Theorem~\ref{thm:cardinality} (note that Theorem~\ref{thm:cardinality} is trivial for $k = 1$).

\begin{theorem}
For $k \geq 2$, the approximation ratio of Algorithm~\ref{alg:DetRandGreedy} is at least $\left(1-\frac{1}{k}\right)^{k-1} \geq 1/e$, and it performs $O(k^2n)$ oracle queries.
\end{theorem}

\begin{proof}
By combining Lemmata~\ref{lem:upper-prob}, \ref{lem:decrease} and~\ref{lem:improve}, we get
\begin{equation}{\textstyle \bE_{S \sim \cD_i}}[f(S)]-{\textstyle \bE_{S \sim \cD_{i - 1}}}[f(S)] \geq \frac{1}{k} \cdot {\textstyle \bE_{S \sim \cD_{i - 1}}}\left[\left(1-\frac{1}{k}\right)^{i-1} \cdot f(OPT)-f(S)\right] \enspace.\label{ineq-improve}\end{equation}
Next, we prove by induction that:
\[{\textstyle \bE_{S \sim \cD_{i}}}[f(S)] \geq \frac{i}{k} \cdot \left(1-\frac{1}{k}\right)^{i-1}f(OPT) \enspace.\]

For $i = 0$, this is true since $f(\varnothing) \geq 0 = (0/k) \cdot (1 - 1/k)^{-1} \cdot f(OPT)$. Assume now that the claim holds for every $i' < i$, let us prove it for $i > 0$.
\begin{align*}
{\textstyle \bE_{S \sim \cD_{i}}}[f(S)] & \geq  {\textstyle \bE_{S \sim \cD_{i-1}}}[f(S)] + \frac{1}{k} \cdot {\textstyle \bE_{S \sim \cD_{i - 1}}}\left[\left(1-\frac{1}{k}\right)^{i-1} \cdot f(OPT)-f(S)\right] \\
& = \left(1- \frac{1}{k}\right) \cdot {\textstyle \bE_{S \sim \cD_{i-1}}}[f(S)] + \frac{1}{k} \cdot \left(1-\frac{1}{k}\right)^{i-1} \cdot f(OPT)\\
& \geq \left(1- \frac{1}{k}\right) \cdot \frac{i-1}{k} \cdot \left(1-\frac{1}{k}\right)^{i-2} \cdot f(OPT) + \frac{1}{k} \cdot \left(1-\frac{1}{k}\right)^{i-1} \cdot f(OPT)\\
& = \frac{i}{k} \cdot \left(1-\frac{1}{k}\right)^{i-1} \cdot f(OPT)
\enspace,
\end{align*}
where the first inequality follows by inequality~\eqref{ineq-improve}, and the second inequality follows by the induction hypothesis.

The approximation ratio guaranteed by the theorem follows immediately by plugging $k$ into the induction hypothesis.
Finally, Observation~\ref{obs:simple} implies $|\cD_i| \leq ik+1$, and thus, in the $i$-th iteration Algorithm~\ref{alg:DetRandGreedy} makes at most $n \cdot \support(\cD_{i - 1}) \leq n \cdot |\cD_{i - 1}| \leq nik$ oracle queries. Thus, the total number of oracle queries in all the iterations is at most $O(k^2 n)$.
\end{proof}

\subsection{A Tight Example for Algorithm~\ref{alg:DetRandGreedy}}\label{sec:badexample}

In this section we give an example of a ``bad'' instance for which Algorithm~\ref{alg:DetRandGreedy} has an approximation ratio of at most $e^{-1} + O(\frac{1}{k})$. Specifically, the optimal solution for the instance we describe has a value of at least $1$, while Algorithm~\ref{alg:DetRandGreedy} may produce a set of value at most $e^{-1} + O(\frac{1}{k})$. In the rest of this section we assume $k$ is larger than some arbitrary constant $\ell$ (to be determined later).

The ground set $\cN$ of our bad instance is the union of two sets $O$ and $Y$, both of size $k$ (if one wishes to have $n > 2k$, it is possible to add an arbitrary number of elements that do not affect the objective function). 
The objective function of the instance is the function $f\colon 2^\cN \to \bR^+$ defined as follows,
\[
	f(S)
	=
	\frac{|S \cap O|}{k} \cdot \left(1 - \frac{|S \cap Y|}{k}\right) + \left[g\left(\frac{|S \cap Y|}{k}\right) + \frac{\ell \cdot |S \cap Y|}{k^2}\right] \cdot \left(1 - \frac{|S \cap O|}{k}\right)
	\enspace,
\]
where $g\colon [0, 1] \to [0, 1]$ is a function given by the following formula:
\[
	g(x)
	=
	\begin{cases}
		(x-1) \cdot \ln(1 - x) & \text{for $0 \leq x \leq 1 - e^{-1}$} \enspace,\\
		e^{-1} & \text{for $1 - e^{-1} \leq x \leq 1$} \enspace.
	\end{cases}
\]
Observe that $g(x)$ is a continuous function. Additionally, we note that,
\[
	g'(x)
	=
	\begin{cases}
		1 + \ln(1-x) & \text{for $0 \leq x \leq 1 - e^{-1}$} \enspace,\\
		0 & \text{for $1 - e^{-1} \leq x \leq 1$} \enspace,
	\end{cases}
\]
and
\[
	g''(x)
	=
	\begin{cases}
		\frac{1}{x-1} & \text{for $0 \leq x < 1 - e^{-1}$} \enspace,\\
		0 & \text{for $1 - e^{-1} < x \leq 1$} \enspace.
	\end{cases}
\]
Observe that $g'$ is always non-negative and $g''$ is always non-positive. Thus, $g$ is a non-decreasing continuous concave function with $g(0)=0$ and $g(1)= e^{-1}$.

It is useful to find expressions for the marginal contribution of an element $u \in \cN$ to a set $S \subseteq \cN$ given the objective $f$. Let $y=|S \cap Y|/k$ and $x=|S \cap O|/k$, then
\begin{align}
f(u \mid S) & = -\frac{x}{k} + \left(1-x\right)\left(g(y+\nicefrac{1}{k})-g(y) +\frac{\ell}{k^2}\right) & \forall u \in Y \setminus S \label{eq:marginal_Y}\\
f(u \mid S) & = \frac{1}{k}\left((1-y)- \left(g(y)+\frac{\ell y}{k}\right)\right) & \forall u \in O \setminus S \label{eq:marginal_O}
\end{align}

\begin{observation}
The function $f$ is a submodular function.
\end{observation}
\begin{proof}
The marginal~\eqref{eq:marginal_Y} of an element $u \in Y$ is a decreasing function of $x$ since $g$ is non-decreasing and of $y$ since $g$ is concave. On the other hand, the marginal~\eqref{eq:marginal_O} of an element $u \in O$ is independent of $x$ and a decreasing function of $y$ since $g$ is non-decreasing.
\end{proof}

In Appendix \ref{app:tight} we complete the proof. We show that given the above bad instance, Algorithm~\ref{alg:DetRandGreedy} may terminate with a distribution over subsets of $Y$. The value of $f$ for any such set $S$ is at most:
\[
	f(S)
	=
	g\left(\frac{|S|}{k}\right) + \frac{\ell \cdot |S|}{k^2}
	\leq
	e^{-1} + \frac{\ell}{k}
	=
	e^{-1} + O\left(\frac{1}{k}\right)
	\enspace.
\]

On the other hand, $O$ is a feasible solution and $f(O) = 1$. Thus, completing the proof.

\section{Conclusions}

In this paper we proposed a new technique for derandomization of algorithms in the area of submodular function maximization.
For unconstrained submodular maximization we showed that randomization is not necessary for obtaining the best possible approximation ratio. For submodular maximization with a cardinality constraint we obtained nearly the best known result.

The main interesting open question is whether algorithms that are based on the multilinear extension can be derandomized. In particular, it is interesting whether the continuous greedy approach \cite{CCPV11,FNS11} used to obtain optimal results for maximizing a monotone submodular function subject to a matroid independence constraint can be derandomized.
One possible direction is to try to approximate the multiliner extension function deterministically using its special properties.
Another interesting question is whether the number of oracle calls of the deterministic algorithms can be reduces. A possible way to speed up algorithms produced by our method is to keep the size of the distributions small by avoiding splitting sets when this results in sets of a too low probability. As long as only low probability sets are affected, this should not significantly decrease the quality of the output, while reducing the number of oracle queries needed (and speeding up the algorithm).

\bibliographystyle{plain}
\bibliography{thesis}

\appendix

\section{Omitted Proofs} \label{app:omitted}

The following claim slightly improves item~\ref{obs:size_increase24} of Observation~\ref{obs:simple_unconstrained}, and shows that every iteration of Algorithm~\ref{alg:doub-greedy} can be implemented in near linear time.

\begin{claim}\label{cl-unconstraind-extreme}
Let $(P)$ be the formulation in Algorithm~\ref{alg:doub-greedy}. There always exists a solution to $(P)$ with at most $1+|\cD_{i-1}|$ non-zero variables. Moreover, this solution can be computed in near linear time.
\end{claim}

 \begin{proof}
To get the desired solution of $(P)$ we need make a few simple manipulations to $(P)$. First, we replace the first constraint of $(P)$ with an objective function asking to maximize:
\[\bE_{\cD_{i - 1}}[z(X,Y)\cdot a_i(X) + w(X,Y)\cdot b_i(Y)] - 2 \cdot \bE_{\cD_{i - 1}}[z(X,Y) b_i(Y)] \enspace,\]
Since $(P)$ is feasible by Observation~\ref{obs:simple_unconstrained}, any optimal solution for the new formulation is a feasible solution for $(P)$. We can simplify the new objective of $(P)$ by removing constants and using the fact that $w(X, Y)$ is fully determined by $z(X, Y)$ due to the equality $z(X,Y)+w(X,Y)=1$. This yields:
\[\bE_{\cD_{i - 1}}[z(X,Y)\cdot \left(a_i(X)-3 b_i(Y)\right)] \enspace.\]
Similarly, by exchanging terms, the second constraint of $(P)$ can be replaced with the equivalent form:
\[\bE_{\cD_{i - 1}}[z(X,Y)\cdot \left(b_i(Y)- 3 \cdot a_i(X)\right)] \leq  \bE_{\cD_{i - 1}}[ b_i(Y)-2\cdot a_i(X)] \enspace.\]

The resulting linear program, for which we need to find an optimal solution, is a variant of the fractional knapsack problem of the following form (where the number $m$ of items is $|\support(\cD_{i - 1})| \leq |\cD_{i - 1}|$).
\[\begin{array}{lll}
	\max & \sum_{j=1}^{m}v_j \cdot z_j \\[2mm]
	& \sum_{j=1}^{m}s_j \cdot z_j &\leq B\\[2mm]
	& 0 \leq z_j \leq 1 & \forall\; 1 \leq j \leq m
\end{array}\]

The only change compared to a standard (fractional) knapsack problem is that $v_j$ and $s_j$ may have negative values (such values can be interpreted as an option to buy additional knapsack space).
However, a simple modification of the, so called, density rule can be used to solve the problem optimally.
First, take to the solution all items with $v_j>0$ and $s_j<0$. Also, omit all items of $v_j<0$ and $s_j>0$.
This leaves us with two types of items: ``positive'' items having $v_j, s_j \geq 0$, and ``negative'' items having $v_j, s_j<0$. We sort the positive items in decreasing order of $v_j/s_j$ (intuitively, the value that we can earn per unit of the knapsack).
Similarly, we sort the negative items increasingly by $v_j/s_j$ (intuitively, the price we need to pay to buy a unit of the knapsack).
The algorithm then starts by (fractionally) taking the first positive items until the knapsack becomes full (or we are out of positive items). 
We then continue (fractionally) taking positive items and negative items in parallel as long as the value per unit gained by the positive item is at least equal to the price per unit paid for the negative item.
It is easy to see that this algorithm produces an optimal solution for the above fractional knapsack problem, and whenever it terminates there is only a single (positive or negative) item taken fractionally.

To complete the proof of the claim observe that when translating a solution for the fractional knapsack problem into a solution for the original formulation $(P)$ we get two non-zero variables for every item taken fractionally, and one non-zero variable for every other item.
\end{proof}

Lemma~\ref{lem:improve} is an immediate implication of Lemma~2.2 of~\cite{BFNS14}. For completeness, we give an independent proof of it.

\begin{replemma}{lem:decrease}
For any subset $T$ and a distribution $\cD$,
\[
	{\textstyle \bE_{S \sim \cD}}[f(T \cup S)] \geq f(S) \cdot \min_{u \in N} {\textstyle \Pr_{S \sim \cD}}[u \not \in S]
	\enspace.
\]
\end{replemma}
\begin{proof}
Let $u_1, u_2, \dotsc, u_n$ denote the elements of $N$ sorted by a non-decreasing order of the probability ${\textstyle \Pr_{S \sim \cD}}[u \not \in S]$. Additionally, let $A_i=\{u_1,u_2, \ldots, u_i\}$ be the set of the first $i$ elements in this order (for every $0 \leq i \leq n$). Then,
\begin{align*}
	&
	{\textstyle \bE_{S \sim \cD}}[f(T \cup S)]
	=
	f(T) + \sum_{i=1}^{n}{\textstyle \bE_{S \sim \cD}}[1[u_i \in S] \cdot f(u_i \mid T \cup (A_{i-1} \cap S))]\\
	\geq{} &
	f(T) + \sum_{i=1}^{n}{\textstyle \bE_{S \sim \cD}}[1[u_i \in S] \cdot f(u_i \mid T \cup A_{i-1})]
	=
	f(T) + \sum_{i=1}^{n} f(u_i \mid T \cup A_{i-1})] \cdot {\textstyle \Pr_{S \sim \cD}}[u_i \in S]\\
	={} &
	{\textstyle \Pr_{S \sim \cD}}[u_1 \not \in S] \cdot f(T) + \sum_{i=1}^{n - 1}({\textstyle \Pr_{S \sim \cD}}[u_i \in S]-{\textstyle \Pr_{S \sim \cD}}[u_{i + 1} \in S]) \cdot f\left(T \cup A_i\right) \\ &+ {\textstyle \Pr_{S \sim \cD}}[u_n \in S] \cdot f(\cN)
	\geq
	{\textstyle \Pr_{S \sim \cD}}[u_1 \not \in S]\cdot f(T)
	\enspace,
\end{align*}
where the first inequality holds by submodularity and the second inequality is based on the fact that ${\textstyle \Pr_{S \sim \cD}}[u_i \in S]$ is a non-increasing function of $i$.
\end{proof} 

\section{Proof of the Tight Example for Algorithm~\ref{alg:DetRandGreedy}}\label{app:tight}

In this section we continue the proof of the tight example for Algorithm~\ref{alg:DetRandGreedy}. Let $f$ be the submodular function defined in Section \ref{sec:badexample}. We analyze a possible execution of the algorithm given this function.

Let us denote the elements of $Y$ by $u_1, u_2, \dotsc, u_k$. For notational convenience, given an index $i > k$, we denote by $u_i$ the element $u_j$ having $i \equiv j \pmod k$. Given a value $z \in [0, 1]$, let us characterize a distribution $\cD(z)$ as follows. The support of the distribution $\cD(z)$ contains at most $2k$ states:
\begin{itemize}
	\item For every $1 \leq i \leq k$, the state $S^L_i = \{u_j\}_{j = i}^{i+ \lfloor k z\rfloor}$ has a probability of $z - \lfloor k z\rfloor/k$.
	\item For every $1 \leq i \leq k$, the state $S^S_i = \{u_j\}_{j = i}^{i+ \lfloor k z\rfloor-1}$ has a probability of $\lfloor kz + 1\rfloor/k - z$.
\end{itemize}
It can be verified that all the above probabilities add up to $1$, and thus, $\cD(z)$ is a valid distribution. Technically the above definition of $\cD(z)$ sometimes defines multiple identical states (for example, the states $S^S_i$ are identical when $z < 1/k$). Whenever this happens, we formally unify these states and give the unified state a probability equal to the sum of the probabilities of the unified states. In the rest of the proof we ignore that possibility for simplicity.

Intuitively, $\cD(z)$ is a distribution over two types of subsets of $Y$: cyclically continuous states of size $1 + \lfloor k z\rfloor$ and cyclically continuous states of size $\lfloor k z\rfloor$. The distribution is symmetrical in the sense that all the cyclically continuous states of a given length have equal probabilities.
\begin{observation}
For every $z \in [0, 1]$ and element $u \in \cN$,
\[
	\Pr_{S \sim \cD(z)}[u \in S]
	=
	\begin{cases}
		z & \text{if $u \in Y$} \enspace,\\
		0 & \text{if $u \in O$} \enspace.
	\end{cases}
\]
\end{observation}
\begin{proof}
It is clear that elements of $O$ never appear in a set distributed like $\cD(z)$, thus, we consider only an element $u \in Y$. By symmetry $u$ appears in $\lfloor k z\rfloor + 1$ cyclically continuous states of size $\lfloor k z\rfloor + 1$ and $\lfloor k z\rfloor$ cyclically continuous states of size $\lfloor k z\rfloor$. Thus,
\begin{align*}
	\Pr_{S \sim \cD(z)}[u \in S]
	={} &
	(\lfloor k z\rfloor + 1) \cdot (z - \lfloor k z\rfloor/k) + (\lfloor k z\rfloor) \cdot (\lfloor kz + 1\rfloor/k - z)\\
	={} &
	\lfloor k z\rfloor \cdot (\lfloor kz + 1\rfloor -  \lfloor k z\rfloor)/k + (z - \lfloor k z\rfloor/k)
	=
	z
	\enspace.
	\qedhere
\end{align*}
\end{proof}

\begin{lemma} \label{lem:possible_run}
For every $0 \leq i \leq k$, Algorithm~\ref{alg:DetRandGreedy} may set $\cD_i = \cD(z_i)$, where $z_i= 1-(1-\nicefrac{1}{k})^i$.
\end{lemma}
\begin{proof}
We prove the lemma by induction. Notice that $\cD(z_0)$ puts all the probability on the empty set, thus, it is identical to $\cD_0$. This completes the proof of the base case. Next, assume that Algorithm~\ref{alg:DetRandGreedy} choosed $\cD_{i - 1} = \cD(z_{i - 1})$ for some $1 \leq i \leq k$, and let us prove that it can end up with $\cD_i = \cD(z_i)$. Let us start with analyzing the marginal contributions of the various elements to a random set from $\cD(z_{i - 1})$.

Let $z'_{i - 1} = \lfloor k z_{i - 1}\rfloor / k$, and note that $z'_{i - 1} \leq z_{i - 1} \leq z'_{i - 1} + \nicefrac{1}{k}$. Consider an arbitrary element $u_o \in O$. Every set $S \in \support(\cD_{i - 1})$ contains either $kz'_{i - 1}$ or $kz'_{i - 1} + 1$ elements of $Y$, and thus, by submodularity, we can lower bound the expected marginal contribution of $u_o$ to a random set of $\cD_{i - 1}$ with its marginal contribution to a set containing $kz'_{i - 1}$ elements of $Y$. By~\eqref{eq:marginal_O} we now get:
\begin{align*}
	\bE_{S \sim \cD_{i - 1}} [f(u_o \mid S)]
	\leq{} &
	\frac{1}{k}\left((1-z'_{i - 1})- \left(g(z'_{i - 1})+\frac{\ell z'_{i - 1}}{k}\right)\right)\\
	\leq{} &
	\frac{1}{k}\left((1-z'_{i - 1})- (z'_{i - 1} - 1)\ln(1 - z'_{i - 1})\right)
	=
	\frac{1}{k}(1-z'_{i - 1}) \cdot (1 + \ln(1 - z'_{i - 1}))
	\enspace.
\end{align*}

On the other hand, consider an element $u_y \in Y$. A random set from $\cD_{i - 1}$ contains $u_y$ with probability $z_{i - 1}$, in which case the marginal contribution of $u_y$ is $0$. Every other set $S$ in the distribution contains either $kz'_{i - 1}$ or $kz'_{i - 1} + 1$ elements of $Y$, and thus, by submodularity, we can upper bound the expected marginal contribution of $u_y$ to such a set with its marginal contribution to a set containing $kz'_{i - 1} + 1$ elements of $y$. By~\eqref{eq:marginal_Y} we now get:
\begin{align}
	\bE_{S \sim \cD_{i - 1}} [f(u_y \mid S)]
	={} &
	(1 - z_{i - 1}) \cdot \left(g(z'_{i - 1}+\nicefrac{2}{k})-g(z'_{i - 1}+\nicefrac{1}{k}) +\frac{\ell}{k^2}\right) \nonumber\\
  \geq{} &
	\frac{1}{k} \cdot (1 - z_{i - 1}) \cdot \left(g'(z'_{i - 1}+\nicefrac{2}{k}) +\frac{\ell}{k}\right) \label{ineq2-ex}\\
  \geq{} &
	\frac{1}{k} \cdot (1 - z'_{i - 1} - \nicefrac{1}{k}) \cdot \left(1 + \ln(1 - z'_{i - 1}-\nicefrac{2}{k}) +\frac{\ell}{k}\right) \nonumber
	\enspace,
\end{align}

where Inequality~(\ref{ineq2-ex}) follows by the concavity of $g$. 
Using the two inequalities we get,

\begin{align}
& \bE_{S \sim \cD_{i - 1}} [f(u_y \mid S)] - \bE_{S \sim \cD_{i - 1}} [f(u_o \mid S)]  \nonumber\\
 \geq{} & \frac{1}{k} \cdot \left(1 - z'_{i - 1}\right)\left(\ln\left(1- \frac{2}{k(1 - z'_{i - 1})}\right)+ \frac{\ell}{k}\right) - \frac{1}{k^2} \cdot \left(1 + \ln(1 - z'_{i - 1}-\nicefrac{2}{k}) +\frac{\ell}{k}\right) \nonumber \\
 \geq{} & \frac{1}{e\cdot k} \left(\ln\left(1- \frac{2e}{k}\right)+ \frac{\ell}{k}\right) - \frac{2}{k^2} \label{ineqq1}\\
 \geq{} & \frac{1}{e\cdot k}\left(-\frac{2e}{k\left(1-\frac{2e}{k}\right)} + \frac{\ell}{k}\right) - \frac{2}{k^2} \label{ineqq2}\\
 \geq{} & \frac{1}{e\cdot k}\left(-\frac{4e}{k} + \frac{\ell}{k}\right) - \frac{2}{k^2} \geq 0 \label{ineqq3}
	\enspace,
\end{align}
where Inequality~(\ref{ineqq1}) follows for a large enough $\ell$ ($\geq 4e$) since $k\geq \ell$ by our assumption and $0\leq z'_{i - 1} \leq 1-e^{-1}$. Inequality~(\ref{ineqq2}) follows by the inequality $\ln(1-y) \geq -\frac{y}{1-y}$, which holds for $y \in [0, 1)$. Finally, Inequality~(\ref{ineqq3}) and the last inequality both follow by considering a large enough $\ell$ ($\geq 6e$) and recalling that $k \geq \ell$.


Since the last inequality holds for every pair of elements $u_o \in O$ and $u_y \in Y$, it implies that the set $M_i$ chosen by the algorithm is exactly $Y$. Given this observation, the formulation $(P)$ of Algorithm~\ref{alg:DetRandGreedy} in iteration $i$ becomes:
\[\begin{array}{lllll}
	(P) & \max & \sum_{u\in Y} \bE_{S \sim \cD_{i - 1}}[x(u, S) \cdot f(u \mid S)] \mspace{-72mu}&& \\
	&& \bE_{S \sim \cD_{i - 1}}[x(u, S)] & \leq \nicefrac{1}{k} \cdot \left(1-z_{i - 1}\right) & \forall u\in Y\\
&& \sum_{u\in M_i} x(u, S) + \ell(S) & = 1 & \forall S \in \support(\cD_{i-1}) \\
&&x(u, S), \ell(S) & \geq 0 & \forall u \in M_i, S \in \support(\cD_{i-1})
\end{array}\]
We need to show that there exists an optimal extreme point solution for this formulation which makes the algorithm set $\cD_i = \cD(z_i)$. There are two cases to consider. If $\cD(z_{i - 1})$ and $\cD(z_i)$ have the same states (\ie, $\lfloor kz_{i - 1} \rfloor = \lfloor kz_i \rfloor$), then Algorithm~\ref{alg:DetRandGreedy} can come up with a solution $x^*$ for $(P)$ assigning $x(u_{j + \lfloor kz_i \rfloor}, S^S_j) = (z_i - z_{i - 1}) / \Pr_{\cD_{i - 1}}[S^S_j]$ for every $1 \leq j \leq k$ and the value $0$ to the other $x$ variables (the values of the $\ell$ variables are induced by the values of the $x$ variables, and thus, we do not state their assignment). The solution $x^*$ is feasible since, for every $1 \leq j \leq k$:
\begin{align*}
	\bE_{S \sim \cD_{i - 1}}[x(u_j, S)]
	={} &
	{\textstyle \Pr_{\cD_{i - 1}}}[S^S_j] \cdot \frac{z_i - z_{i - 1}}{\Pr_{\cD_{i - 1}}[S^S_j]}
	=
	[1-(1-\nicefrac{1}{k})^i] - [1-(1-\nicefrac{1}{k})^{i - 1}]\\
	={} &
	(1-\nicefrac{1}{k})^{i - 1}[1 - (1 - \nicefrac{1}{k})]
	=
	\nicefrac{1}{k} \cdot (1 - z_{i - 1})
	\enspace.
\end{align*}
It can be checked that $x^*$ indeed leads Algorithm~\ref{alg:DetRandGreedy} to set $\cD_i = \cD(z_i)$. To see that $x^*$ is optimal, notice that it adds elements only to the smaller sets (which results in a larger marginal gain by submodularity), and it adds every element to the maximal extent allowed by the first type of constraints. Finally, to see that $x^*$ is an extreme point solution notice that it is the only solution maximizing the objective function $c \cdot x$, where $c$ is a vector taking the value $1$ exactly in the coordinates for which $x^*$ is non-zero.

The second case we need to consider is when $\cD(z_{i - 1})$ and $\cD(z_i)$ have different states (\ie, 1 + $\lfloor kz_{i - 1} \rfloor = \lfloor kz_i \rfloor$). In this case Algorithm~\ref{alg:DetRandGreedy} can come up with a solution $x^*$ for $(P)$ assigning $x(u_{j + \lfloor kz_{i - 1} \rfloor}, S^S_j) = 1$ and $x(u_{j + \lfloor kz_i \rfloor}, S^L_j) = k^{-1}(kz_{i - 1} - z_{i - 1} - \lfloor kz_{i - 1} \rfloor) / \Pr[S^L]$ for every $1 \leq j \leq k$ and the value $0$ to the other $x$ variables. The solution $x^*$ is feasible since, for every $1 \leq j \leq k$:
\begin{align*}
	\bE_{S \sim \cD_{i - 1}}[x(u_j, S)]
	={} &
	{\textstyle \Pr_{\cD_{i - 1}}}[S^S_j] + {\textstyle \Pr_{\cD_{i - 1}}}[S^L_j] \cdot \frac{kz_{i - 1} - z_{i - 1} - \lfloor kz_{i - 1} \rfloor}{k \cdot \Pr_{\cD_{i - 1}}[S^L_j]}\\
	={} &
	\lfloor kz_{i - 1} + 1\rfloor/k - z_{i - 1} + \frac{kz_{i - 1} - z_{i - 1} - \lfloor kz_{i - 1} \rfloor}{k}
	=
	\nicefrac{1}{k} \cdot (1 - z_{i - 1})
	\enspace.
\end{align*}
It can be checked that $x^*$ again leads Algorithm~\ref{alg:DetRandGreedy} to set $\cD_i = \cD(z_i)$. To see that $x^*$ is optimal, notice that it adds as much as possible elements to the smaller sets (which results in a larger marginal gain by submodularity), and only the remaining capacity given by the first type of constraints is used to add elements to the larger sets. Finally, to see that $x^*$ is an extreme point solution notice that it is the only solution maximizing the objective function $c \cdot x$, where $c$ is a vector taking the value $k + 1$ in the coordinates for which $x^*$ is $1$ and the value $1$ in the other coordinates for which $x^*$ is non-zero.
\end{proof}

To complete the analysis of our bad instance, notice that $O$ is a feasible solution and $f(O) = 1$. On the other hand, Lemma~\ref{lem:possible_run} shows that Algorithm~\ref{alg:DetRandGreedy} may terminate with a distribution over subsets of $Y$. The value of $f$ for any such set $S$ is at most:
\[
	f(S)
	=
	g\left(\frac{|S|}{k}\right) + \frac{\ell \cdot |S|}{k^2}
	\leq
	e^{-1} + \frac{\ell}{k}
	=
	e^{-1} + O\left(\frac{1}{k}\right)
	\enspace.
\]

\end{document}